\documentclass[a4paper,onecolumn,12pt]{edtarticle}

\usepackage{fancyhdr}
\usepackage{amsmath,amsthm,amssymb}
\usepackage{amsthm}
\usepackage{lscape,setspace}
\usepackage{enumerate}
\usepackage{epsfig}
\usepackage{authblk}
\usepackage{longtable}  
\usepackage{booktabs}   
\usepackage{subcaption}
\usepackage{multirow}
\usepackage{algorithm}
\usepackage{algpseudocode}
\usepackage{xcolor}
\usepackage{url}
\newtheorem{theorem}{Theorem}[section]

\newtheorem{lemma}[theorem]{Lemma}

\theoremstyle{definition}
\newtheorem{definition}{Definition}

\fancypagestyle{first}{%
  \fancyhf{}
\fancyhead[R]{\includegraphics[scale=1]{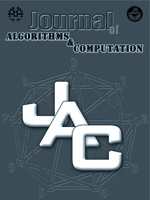}}

\fancyhead[C] { \large  \textbf{Journal of Algorithms and Computation}  \\ \vspace{10pt} \small journal homepage: http://jac.ut.ac.ir}

\fancyhead[L]{\includegraphics[scale=0.183] {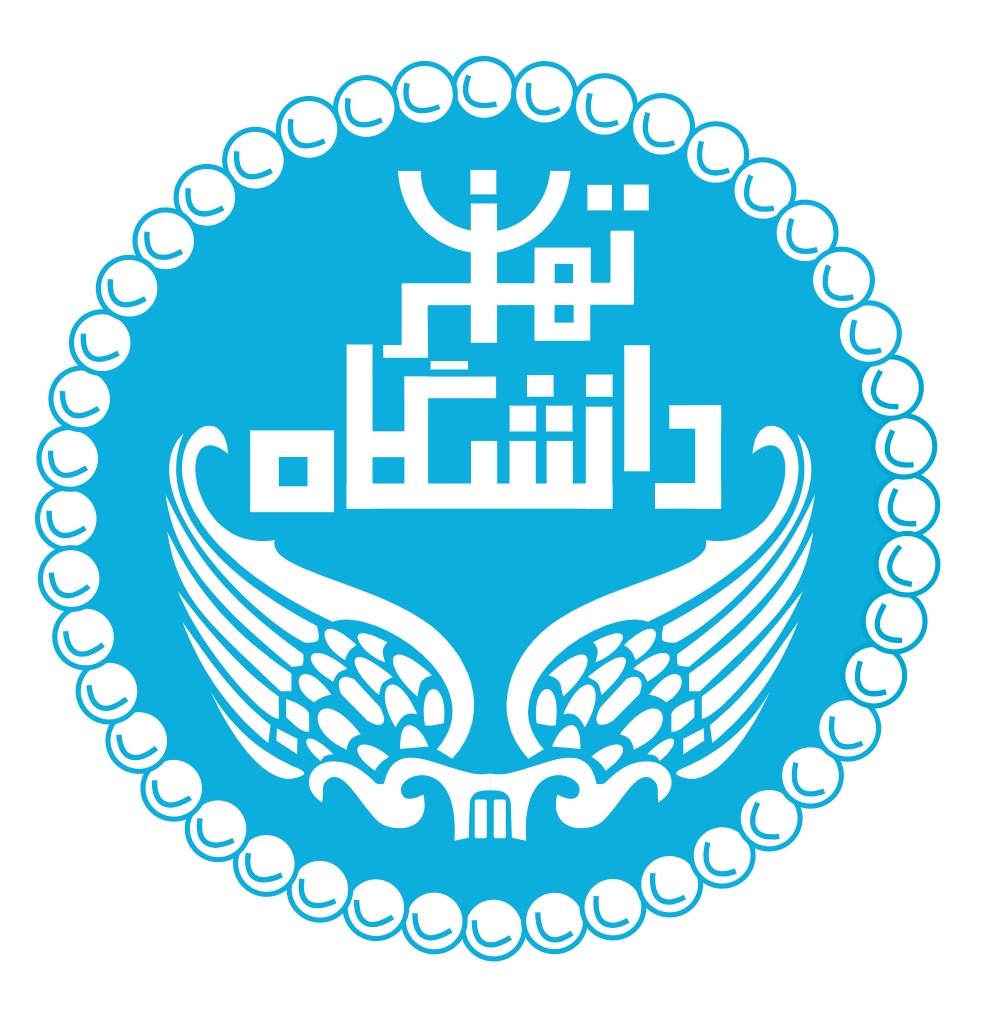}}

\cfoot{Journal of Algorithms and Computation XX issue X, XXX 2024, PP. XX--XX}
}

\fancypagestyle{otherstyle}{%
  \fancyhf{}
\fancyhead[C]{Diameter in Temporal Networks/ JAC XX issue X, XXX 2024, PP. XX--XX}
\fancyhead[L]{\thepage}
}
\title{\vspace{16pt} Modelling and Study of $\tau$ , Peak and Effective 
Diameter in Temporal Networks}
\author[1]{Zahra Farahi\thanks{zhr.farahi@ut.ac.ir}}
\author[1]{Ali Kamandi\thanks{Corresponding author: Ali Kamandi. kamandi@ut.ac.ir }}

\affil[1]{\small School of Engineering Science , College of Engineering, University of Tehran, Tehran, Iran}

\author[1]{Ali Moeini\thanks{moeini@ut.ac.ir}}
\begin{document}

\maketitle
\thispagestyle{first}
\setlength{\parindent}{0pt}
\begin{tabular}{p{285pt}p{5pt}p{160pt} l c l }
\hline
\\
\textbf{ABSTRACT}  & & \textbf{ARTICLE INFO}\\  \cline{1-1}  \cline{3-1}
\vspace{0.2cm}

Understanding how information, diseases, or influence spread across networks is a fundamental challenge in complex systems. While network diameter has been extensively studied in static networks, its definition and behavior in temporal networks remain underexplored due to their dynamic nature. In this study, we present a formal mathematical framework for analyzing diameter in temporal networks and introduce three time-aware metrics: {Effective Diameter} ($\backsim$$\mathcal{D}$), {Peak Diameter} ($* \mathcal{D}$), and {$\tau$-Diameter} ($\tau \mathcal{D}$), each capturing distinct temporal aspects of connectivity and diffusion.

Our approach combines theoretical analysis with empirical validation using four real-world datasets: high school, hospital, conference, and workplace contact networks. We simulate flow propagation on temporal networks and compare the observed diameters with the proposed theoretical Equations. Across all datasets, our model demonstrates high accuracy, with low RMSE and absolute error values. Furthermore, we observe that the effective diameter decreases with increasing average degree and increases with network size. The results also show that $\tau \mathcal{D}$ and $* \mathcal{D}$ are more sensitive to node removal, highlighting their relevance for applications such as epidemic modeling.

By bridging formal modeling and empirical data, our framework offers new insights into the temporal dynamics of networked systems and provides tools for assessing robustness, controlling information spread, and optimizing interventions in time-sensitive environments.

 & & \vspace{0.2cm}
\begin {tabular}{p{160pt}l}
\textit{Article history:}  \\
\small Research paper\\
\small Received XX, XXX XXXX\\
\small Accepted XX, XXX XXXX \\
\small Available online XX, XXX XXXX\\
\end{tabular}
\\
\cline{1-1}
\vspace{0.2cm} \small \textit{Keyword:} Temporal networks, Scale-free, Network Diameter, Shortest path. & & \\
\end{tabular}

\setlength{\parindent}{0pt}

\line(1,0){485}


\section{Introduction}

Complex networks have emerged as a cornerstone of modern research, offering powerful frameworks for analyzing diverse systems, from biological and social interactions to technological and information networks. These networks capture the intricate relationships between entities, helping researchers understand the underlying dynamics of connectivity, diffusion, and robustness. Among the many metrics used to characterize complex networks, diameter plays a fundamental role, especially in epidemic modeling, network security, and information spread. 

The study of network diameter dates back to classical graph theory, where early approaches relied on breadth-first search (BFS) to determine the longest shortest path between nodes \cite{64newman2001scientific}. Over time, alternative methods, such as random walk algorithms \cite{63lee2008random} and spectral approaches using adjacency matrices and eigenvalues \cite{5cvetkovie1979spectru, 6chung1989diameters, 7delorme1991diameter}, were introduced to refine these calculations. While these techniques have provided valuable insights, they primarily focus on static networks, where connections remain fixed over time. However, in real-world systems, most networks are inherently temporal which connections appear, disappear, and evolve dynamically. 

One of the most striking properties of real-world networks is their scale-free structure, where the majority of nodes have few connections, but a small number of highly connected nodes play a disproportionate role in connectivity \cite{60dorogovtsev2008critical}. This structure contributes to the small-world effect, where distances between nodes grow logarithmically rather than linearly as network size increases \cite{59samoylenko2023there}. These findings have profound implications for the spread of diseases, information, and even cyber threats, making the study of diameter crucial for designing efficient intervention strategies. 

By understanding the diameter of a network, we can predict how quickly a virus spreads and use this information to design control actions such as quarantine and isolation. During the COVID-19 pandemic, researchers leveraged network analysis to model the speed of transmission based on shortest path dynamics \cite{38sun2020did, 51jo2021social}. Similarly, in the realm of social influence and misinformation, identifying key spreaders, whether through random walk-based influence detection \cite{14zhao2022random}, multi-local dimension analysis \cite{34wen2020vital}, or k-shell decomposition \cite{40wang2020identifying}, has become essential for combating misinformation and optimizing content dissemination. 

Despite extensive research on static networks, the field has largely overlooked a critical reality: real-world networks are not static, but temporal. In many systems, from communication and transportation networks to social interactions, connections fluctuate over time. This realization has led to the emergence of temporal networks, where the timing and duration of connections play a crucial role in determining network properties \cite{52chen2020time, 54zhang2019flow, 55song2020spatial, 56tang2020predictability}. 

The transition from static to temporal networks raises new challenges. In static networks, diameter is well-defined as the longest shortest path between any two nodes. However, in temporal networks, where edges appear and disappear over time, this definition no longer holds. A node may be reachable from another only at certain times, meaning that connectivity is not just a function of network topology, but also of temporal dynamics. 

To address this complexity, researchers have proposed new frameworks. The stochastic shortest path model introduced by Andreatta et al. accounts for random variations in edge weights and connection durations \cite{4andreatta1988stochastic}. More recently, Pedreschi et al. introduced the concept of the temporal rich-club phenomenon, showing that highly connected nodes maintain frequent, recurring interactions, which significantly influence network dynamics \cite{62pedreschi2022temporal}. However, despite these advances, the diameter of temporal networks remains largely unexplored. 

Also, Smith and Doe provide a comprehensive review of the dynamics of social networks, emphasizing the critical role that time plays in the spread of information and social influence \cite{1smith2024temporal}. This understanding is further enriched by Kim and Chen, who introduce a discrete-time framework for modeling epidemic spread, highlighting the impact of time-varying connections on infectious disease dynamics \cite{2kim2024modeling}. Similarly, Garcia and Thompson explore how these temporal dynamics can enhance cybersecurity protocols, proposing adaptive measures that respond to changing network structures \cite{3garcia2025leveraging}. Additionally, Alvarez and Patel investigate resilience in temporal transportation networks, demonstrating how varying connectivity affects operational efficiency and reliability \cite{4alvarez2025understanding}. Finally, Li and Kumar delve into the interplay between network structure and temporal patterns, focusing on influencer dynamics and their implications for marketing strategies \cite{5li2024exploring}. Together, these studies contribute to a nuanced understanding of temporal networks and underscore the necessity of accounting for temporal factors in network analysis.

This study presents a structured method for measuring diameter in temporal networks by focusing on how connectivity changes over time. We define three key measures, effective diameter, peak diameter, and $\tau$-diameter, each reflecting different characteristics of temporal paths and interactions.

To demonstrate the practical value of our method, we analyze four real-world datasets involving high school interactions, hospital contact networks, workplace communications, and conference settings. Through this analysis, we show how temporal features influence network behavior and connectivity patterns.

By combining theoretical modeling with real data, the study offers a clearer picture of how time-dependent interactions shape the structure and flow of information in complex systems. These insights can support a wide range of applications, from controlling the spread of diseases and enhancing cybersecurity to improving transportation systems and understanding social behavior.

\section{Proposed Approach: Modeling Diameter in Temporal Networks}

This section introduces our novel approach to understanding and quantifying the diameter in temporal networks. Unlike traditional static analyses, our method explicitly accounts for the time-dependent nature of connections.

\subsection{Theoretical Framework and New Definitions}

In static networks, a path is defined as a sequence of connected edges, each sharing a common node with the previous one. In temporal networks, however, both the sequence of edges and their chronological order are important.

To provide a clearer understanding of temporal networks, we define a key concept: $flow$. Flow starts from a source node and spreads to its neighbors at each time step ($t$) based on the Breadth-First Search (BFS) algorithm, but only if the connections to those neighbors are active at that moment.

\begin{definition}
A temporal path between nodes $v_i$ and $v_j$, denoted as $\mathcal{P}(i, j)$, is a sequence of edges arranged in chronological order based on their active times. A flow starting from $v_i$ can reach $v_j$ if there exists such a path.

\begin{equation}\label{tpatheq}
  \mathcal{P}(i, j) =\{(v_i ,v_{k1} ,t_1 ), (v_{k_1} ,v_{k2} ,t_2 ), ... , (v_{k_m} ,v_{kj} ,t_m )\}
\end{equation}
where $t_1<t_2<...<t_m$, ensuring that the edges follow a chronological order.
\end{definition}


  \begin{figure}
    \centering
    \begin{subfigure}[b]{0.25\textwidth}
        \centering
        \includegraphics[width=1.2\textwidth]{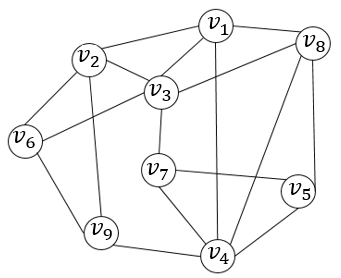}
        \caption{}
        \label{sampleNetworkfig}
    \end{subfigure}
    \hfil
    \begin{subfigure}[b]{0.55\textwidth}
        \centering
        \includegraphics[width=1.2\textwidth]{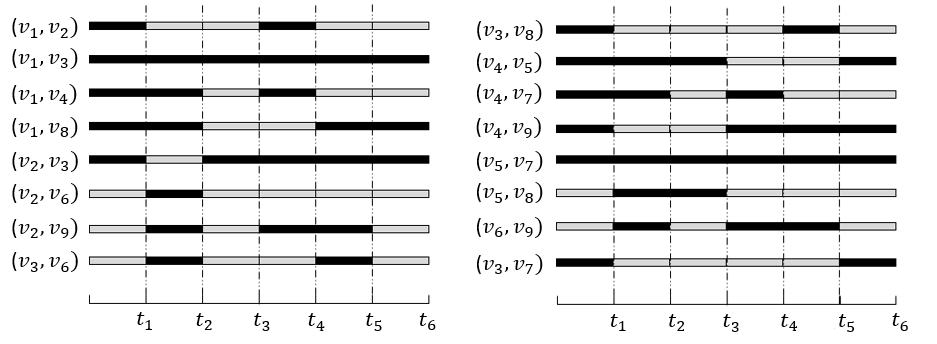}
        \caption{}
        \label{connectiontimefig}
    \end{subfigure}
    \caption{A nine-node network with changing connections over time. (a) Static view, (b) Temporal connections at each time step.}
    \label{samplefig}
  \end{figure}

\begin{definition}
For any node $v_i$, the reachable set $\mathcal{R}_i$ includes all nodes $v_j$ that can be reached from $v_i$ through a temporal path $\mathcal{P}(i, j)$ within $T$ time steps.
\begin{equation}\label{rseteq}
  \mathcal{R}_i =\{ v_j |  \exists\mathcal{P}(i, j) \}
\end{equation}

\end{definition}
\label{def:reachable_Set}
Figure \ref{sampleNetworkfig} shows a network with 9 nodes, while Figure \ref{connectiontimefig} illustrates the temporal connections between nodes, where the connections change over six time steps ($\{t_1,...,t_6\}$).

Figure  \ref{sreacheInTimefig} shows the flow path starting from node $v_6$ in Figure \ref{samplefig}. As seen in the figure, the last node visited by the flow is $v_5$, while $v_7$  is not visited. Therefore, the reachable set of $v_6$ includes all nodes that the flow can reach, except for $v_7$, which remains unreachable.

 Applying Definition \ref{def:reachable_Set}, we obtain:
  \begin{itemize}
    \item $\mathcal{R}_i(t_0) = \{v_6\}$
    \item $\mathcal{R}_i(t_1) = \{v_6\}$
    \item $\mathcal{R}_i(t_2) = \{v_2, v_3, v_6\}$
    \item $\mathcal{R}_i(t_3) = \{v_1, v_2, v_3, v_6\}$
    \item $\mathcal{R}_i(t_4) = \{v_1, v_2, v_3, v_4, v_6,v_9\}$
    \item $\mathcal{R}_i(t_5) = \{v_1, v_2, v_3, v_4, v_6, v_8,v_9\}$
    \item $\mathcal{R}_i(t_6) = \{v_1, v_2, v_3, v_4, v_5, v_6, v_8, v_9\}$
  \end{itemize}
  
Thus, after six time steps, the reachable set of node $v_6$ is $ \{v_1, v_2, v_3, v_4, v_5, v_6, v_8, v_9\}$.

\begin{figure}
  \centering
  \includegraphics[width=7cm]{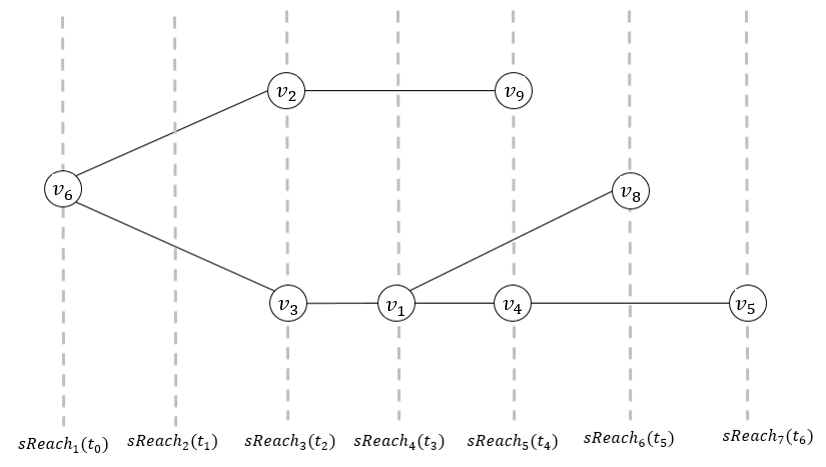}
  \caption{Flow propagation from node $v_6$ in the network of Figure \ref{connectiontimefig}, illustrating the reachable set at each time step.}
  \label{sreacheInTimefig}
\end{figure}

In a temporal random network, the neighbors of a node $v_i$ change over time. At each time step $t$ any two nodes are connected with probability $\hat{p}$. This probability is defined based on the active time $(\zeta)$ and the total observation period $({T})$ as follows:

\begin{equation}\label{peq}
  \hat{p} = \frac{\zeta}{{T}}
\end{equation}

\begin{definition}
The effective degree in temporal random networks is defined with respect to the degree in static random networks as follows:
\begin{equation}\label{eq_eff_degree}
    \langle \hat{k} \rangle = \langle k \rangle \times \frac{\zeta}{{T}}
\end{equation}
where $\langle k \rangle$ represents the average degree \cite{28barabasi2013network}.
\end{definition}

\begin{definition}
  \textbf{Effective Diameter} ($\backsim$$\mathcal{D}$): represents the maximum number of time steps required for a flow to reach all reachable nodes in the network. If the network is disconnected, the effective diameter is determined by the longest temporal path within the largest connected components, considering parallel flow processes occurring in different clusters. We define the effective diameter as:

  \begin{equation}\label{eq_eff_dim}
   \backsim\mathcal{D} = \max_{v_i \in V} \max_{v_j \in \mathcal{R}_i} d_T(v_i, v_j)
\end{equation}
  where \( d_T(v_i, v_j) \) is the shortest temporal path length (in time steps) from \( v_i \) to \( v_j \).
\end{definition}

\begin{definition}
  \textbf{Peak Diameter} ($*$$\mathcal{D}$):  is the time step at which the maximum number of nodes have been reached by a flow originating from any node. This metric captures the point in time when the network is most connected in terms of information spread. Formally, we express it as:  
     \begin{equation}\label{eq_peak_dim}
   * \mathcal{D} = \arg\max_{t \in [t,T]} |\mathcal{R}_t|
\end{equation}

\end{definition}

\begin{definition}
  \textbf{$\tau$-Diameter} ($\tau\mathcal{D}$): is the time step at which one-third of the nodes in the network have been reached by a flow. It represents the point in time when a significant portion of the network has been covered by the flow.  
  
    \begin{equation}\label{eq_ta_dim}
  \tau D = \min_{t \in [1,T]} \left\{ t \mid \left| \mathcal{R}_t \right| \geq \frac{1}{3} |V| \right\}
\end{equation}

\end{definition}

Based on previous definitions, the diameter of a temporal network corresponds to the time required to visit all nodes. If $N$ is the total number of nodes in the network and at each time step a fraction of the nodes, given by $di/dt$, are visited, then we have:

\begin{equation}\label{didt}
  \frac{di}{dt} = \langle k \rangle (N - i) i
\end{equation}

where $(N - i)$ represents the number of unvisited nodes at time $t$, and $i$ is the cumulative number of visited nodes. The time step at which the rate of new visits $di/dt$ reaches its peak is what we defined as the $peak$ $diameter$.

Additionally, the parameter $\tau\mathcal{D}$ represents the time required to visit one-third of the total nodes in the static network. It is given by:

\begin{equation}\label{taueq}
  \tau = \frac{\langle k \rangle}{\beta (\langle k^{2} \rangle - \langle k \rangle)}
\end{equation}

where $\beta$ is a network-dependent parameter that captures the variability in node connectivity.

\subsection{Analytical Modeling}\label{temporalNetworkDiameter}

In this section, we present the theoretical approach for analyzing flow propagation in temporal networks. This analysis helps us gain a deeper understanding of information spread and the time required for full network coverage.

\begin{theorem}\label{lemastochasticdiameter}
Considering $\langle \hat{k} \rangle$ in Equation \ref{eq_eff_degree}, the size of the reachable set at step $t$ in a random temporal network is given by:

\begin{equation}\label{SrachDynamicFormulaeq}
  |\mathcal{R}_t(i)| = 
  \begin{cases}
    \langle \hat{k} \rangle, & \text{if } t = 1 \\[8pt]
    \sum\limits_{l = 1}^{|\mathcal{R}_{t-1}(i)|} 
    \left( \langle \hat{k} \rangle \cdot 
    \frac{n - \sum\limits_{t'=0}^{t-1} |\mathcal{R}_{t'}(i)| - \sum\limits_{z=0}^{l} (z \cdot \langle \hat{k} \rangle)}{n} 
    \right), & \text{otherwise}.
  \end{cases}
\end{equation}

\begin{proof}
We model the spreading process in a random temporal network with $n$ nodes, where each node establishes $\langle \hat{k} \rangle$ random temporal links per time step. Let $\mathcal{R}_t(i)$ denote the set of nodes reachable from node $i$ after $t$ time steps, via temporal paths.

At $t=1$, the initial set of reachable nodes from node $i$ consists of its immediate neighbors. Since the network is random and node degrees follow an average $\langle \hat{k} \rangle$, we have:
\[
|\mathcal{R}_1(i)| = \langle \hat{k} \rangle.
\]

For $t>1$, each node in $\mathcal{R}_{t-1}(i)$ attempts to reach new nodes, selected uniformly at random among the yet-unreached nodes. The probability that a connection from a node in $\mathcal{R}_{t-1}(i)$ leads to a new node depends on the remaining number of unvisited nodes.

Let $N_t = n - \sum_{t'=0}^{t-1} |\mathcal{R}_{t'}(i)|$ denote the number of remaining nodes at step $t$. The $l$-th node in $\mathcal{R}_{t-1}(i)$ has $\langle \hat{k} \rangle$ new edges, but due to possible overlap, the actual number of new nodes it reaches is scaled by the probability of choosing an unreached node, which decreases with each connection attempt. Hence, the expected number of new nodes reached by the $l$-th node is:

\[
\langle \hat{k} \rangle \cdot \frac{N_t - \sum_{z=0}^{l} (z \cdot \langle \hat{k} \rangle)}{n}.
\]

Summing over all $l = 1$ to $|\mathcal{R}_{t-1}(i)|$ yields the total expected number of newly reached nodes at time $t$, which completes the recurrence relation in Equation~\ref{SrachDynamicFormulaeq}.

Note that we assume no backward connections to previous time layers, as such connections would imply the existence of shorter paths, violating the assumption of minimal-length temporal paths. Therefore, the expansion process is strictly layered and forward.
\end{proof}
\end{theorem}

\begin{lemma} \label{lem2}The time required to reach one-third of the nodes in temporal network (${\tau}\mathcal{D}$), can be estimated using: 
\begin{equation} \tau\mathcal{D} \approx \frac{\ln (N/3)}{\ln (1 + \langle \hat{k} \rangle / N)} \end{equation} where ${\tau}\mathcal{D}$ represents the number of steps needed to reach one-third of the network. \end{lemma}

\begin{proof} Using the recursive equation for $|\mathcal{R}_t(i)|$, we approximate its continuous growth with an exponential model: \begin{equation} |\mathcal{R}_t(i)| \approx N \left(1 - e^{- t \cdot \langle \hat{k} \rangle / N} \right) \end{equation} Setting $|\mathcal{R}_t(i)| = N/3$ and solving for $t$ gives the stated result. \end{proof}

These results provide an analytical framework for understanding the spread dynamics in temporal networks and estimating the necessary steps for full network coverage.

\begin{lemma}\label{inversRel}
For a temporal network with a fixed degree distribution and activation time, the expected effective diameter ($\backsim$$\mathcal{D}$) is inversely related to the probability $\hat{p} = \frac{\zeta}{\mathcal{T}}$.
\end{lemma}

\begin{proof}

according to lemma \ref{lem2} we have $\backsim$$\mathcal{D} \approx \frac{\ln (N)}{\ln (1 + \langle \hat{k} \rangle / N)} $ and  $ |\mathcal{R}_t(i)| \approx N \left(1 - e^{- t \cdot \langle \hat{k} \rangle / N} \right) $ Substituting $ |\mathcal{R}_t(i)| = N$, we solve for $t$

\begin{equation} 
\backsim\mathcal{D} \approx \frac{\ln (N)}{\ln (1 +\hat{p} \langle \hat{k} \rangle / N)}
 \end{equation}
Since $\hat{p} = \frac{\zeta}{\mathcal{T}}$ we can express $\backsim \mathcal{D}$ as
\begin{equation}
\backsim\mathcal{D} \approx \frac{\ln N}{\ln(1 + (\zeta / T) \cdot \langle \hat{k} \rangle / N)}
\end{equation}
Clearly, as $\hat{p}$ (or equivalently, $\zeta$) increases, the denominator grows, making $\backsim\mathcal{D}$ smaller.

 Since $\ln(1 + x)$ is an increasing function, a higher $\hat{p}$ leads to a larger denominator and thus a smaller $\backsim\mathcal{D}$. Therefore, the temporal diameter $\backsim\mathcal{D}$ decreases as the effective reachability probability $\hat{p}$ increases.
 
 \begin{equation}
 \backsim\mathcal{D} \approx \frac{\ln N}{\ln(1 + \hat{p} \cdot \langle\hat{k} \rangle / N)}
 \end{equation}

\end{proof}

\begin{theorem}
In a temporal network with a fixed average degree $\langle\hat{k} \rangle$, the effective diameter $\backsim\mathcal{D}$ grows logarithmically with the number of nodes N, assuming uniform temporal activation:

 \begin{equation}
 \backsim\mathcal{D} \propto log N
 \end{equation}

\end{theorem}

\begin{proof}
As derived in previous lemmas, the number of reachable nodes at time t follows the recurrence:
\begin{center}
$
|\mathcal{R}_t(i)| = |\mathcal{R}_{t-1}(i)| + \Delta \mathcal{R}_t.
$
\end{center}
where the newly reached nodes at each step $( \Delta \mathcal{R}_t )$ depend on the previously reached nodes and the effective link activation probability:
\begin{center}
$\Delta \mathcal{R}_t = \hat p \cdot \langle k \rangle \cdot |\mathcal{R}_{t-1}(i)|.$
\end{center}
Using the {effective reachability probability} $( \hat p = \frac{\zeta}{T} )$, we rewrite the growth equation as:

\begin{center}
$|\mathcal{R}_t(i)| = |\mathcal{R}_{t-1}(i)| + \frac{\zeta}{T} \cdot \langle \hat k \rangle \cdot |\mathcal{R}_{t-1}(i)|.$
\end{center}
Approximating the growth as a continuous process, we express the evolution of the reachable set as a differential equation:

\begin{center}
$\frac{d|\mathcal{R}_t(i)|}{dt} = \frac{\zeta}{T} \cdot \langle \hat k \rangle \cdot |\mathcal{R}_t(i)|.$
\end{center}
Solving this equation gives an exponential growth pattern:

\begin{center}
$|\mathcal{R}_t(i)| = N \left(1 - e^{-\frac{\zeta}{T} \cdot \langle \hat k \rangle \cdot t}\right).$
\end{center}
To find the time required to reach the entire network, we set $( |\mathcal{R}_t(i)| = N)$ and solve for t:

\begin{center}
$1 - e^{-\frac{\zeta}{T} \cdot \langle \hat k \rangle \cdot t} = 1.$
\end{center}
which yields:

\begin{center}
$\backsim \mathcal{D}= \frac{T}{\zeta \langle \hat k \rangle} \ln N.$
\end{center}
Thus:

\begin{center}
$\backsim \mathcal{D} \propto \log N.$
\end{center}

\end{proof}

\subsection{Empirical Validation and Simulation}

In this section, we evaluate the accuracy of our proposed model through extensive simulations across various network configurations. Our primary objective is to compare the theoretical diameters predicted by the model with those observed in simulated networks. Specifically, we examine three key aspects: accuracy, the impact of degree and network size, and the behavior of different degree distributions.

To assess the accuracy of Equation (\ref{SrachDynamicFormulaeq}), we implemented a temporal network simulation incorporating realistic degree distributions. A detailed explanation of the simulation algorithm is provided in Section \ref{appendix}.

Figure \ref{NormalDynamicCMPfigk} presents the network diameter for a fixed network size $N=500$ while varying the average degree between 10 and 70. As expected, increasing the average degree leads to a reduction in the effective diameter. When the average degree surpasses 70, the network becomes sufficiently dense, allowing information or flow to propagate to all nodes within just two steps, effectively reducing the diameter to 2. This observation aligns with the claim of Lemma \ref{inversRel}, confirming that the effective diameter and the average degree exhibit an inverse relationship. Moreover, Figure \ref{NormalDynamicCMPfigk} demonstrates that this inverse relationship holds consistently across different degree distributions, indicating that the choice of distribution does not alter this fundamental trend.

\begin{figure}[htbp]
  \centering
  \begin{subfigure}{0.32\textwidth}
    \centering\includegraphics[width=\textwidth]{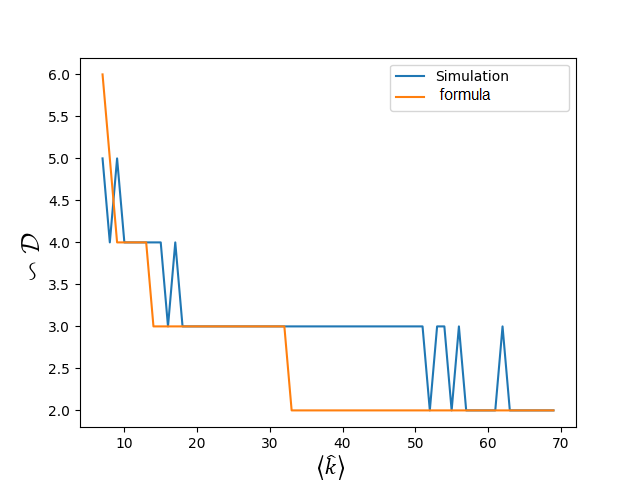}
    \caption{ }\label{DynamicNoramlKfig}
  \end{subfigure}
  \hfill
  \begin{subfigure}{0.32\textwidth}
    \centering\includegraphics[width=\textwidth]{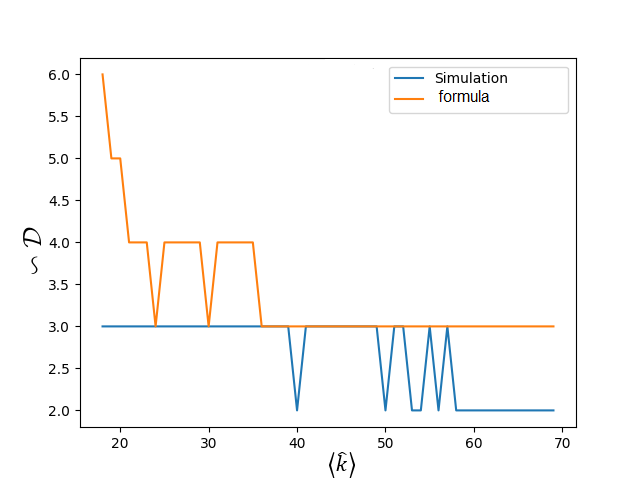}
    \caption{ }\label{DynamicParretoKfig}
  \end{subfigure}
  \hfill
  \begin{subfigure}{0.32\textwidth}
    \centering\includegraphics[width=\textwidth]{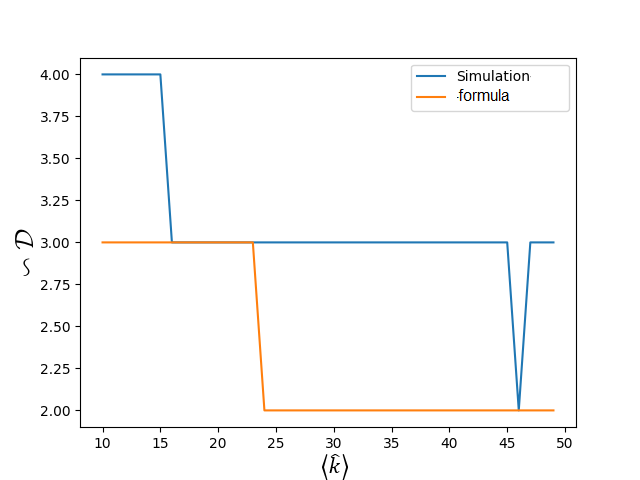}
    \caption{ }\label{DynamicPission-kfig}
  \end{subfigure}
  \caption{Comparison of effective diameter ($\backsim \mathcal{D}$) between theoretical predictions and simulations for different distributions: (a) normal, (b) Pareto, and (c) Poisson. All simulations are conducted with \(N=500\) nodes, and the plots depict the diameter across various values of the average degree 
$\langle \hat k \rangle$.}\label{NormalDynamicCMPfigk}
\end{figure}

Following the accuracy assessment of Equation (\ref{SrachDynamicFormulaeq}), Figure \ref{NormalDynamicCMPfigN} examines the network diameter across different network sizes while maintaining a fixed degree distribution (Normal). The results show that the theoretical predictions closely align with the empirical simulation values, further validating the accuracy of our model. Additionally, as the network size increases, the effective diameter grows accordingly, reinforcing the expected structural properties of large-scale networks.

\begin{figure}[htbp]
  \centering
  \begin{subfigure}{0.32\textwidth}
    \centering\includegraphics[width=\textwidth]{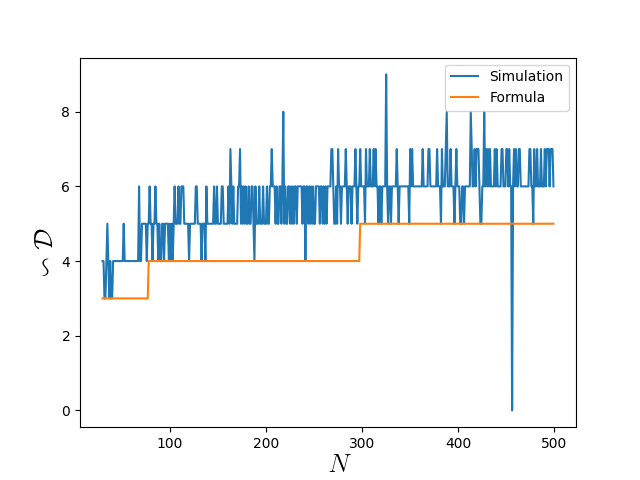}
    \caption{ }\label{DynamicNoramlNfig}
  \end{subfigure}
  \begin{subfigure}{0.32\textwidth}
    \centering\includegraphics[width=\textwidth]{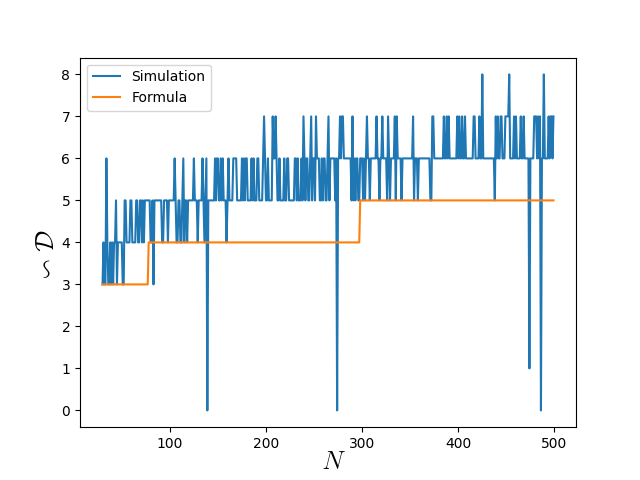}
    \caption{ }\label{ParretoDynamicNfig}
  \end{subfigure}
  \begin{subfigure}{0.32\textwidth}
    \centering\includegraphics[width=\textwidth]{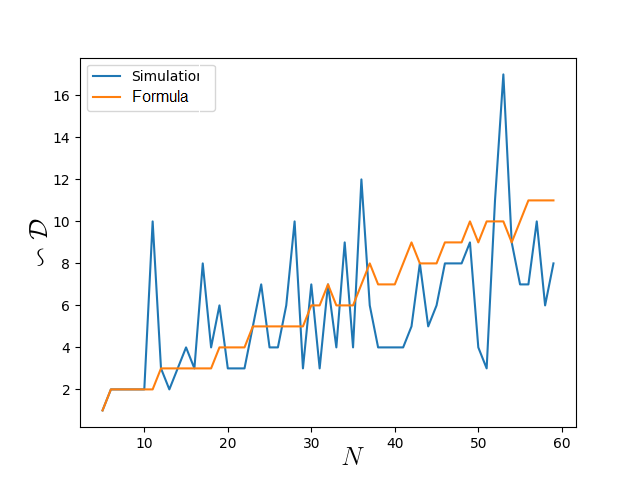}
    \caption{ }\label{DynamicPoissionNfig}
  \end{subfigure}
  \caption{Comparison between the network effective diameter ($\backsim \mathcal{D}$) in simulations and the theoretical equation for different link distributions. The plots correspond to (a) normal, (b) Pareto, and (c) Poisson distributions. In these networks, \(\langle \hat k \rangle = 5\), and the plots illustrate the diameter for varying values of \(N\).}
  \label{NormalDynamicCMPfigN}
\end{figure}

Table \ref{errorComparisontb} provides a quantitative comparison between the theoretical predictions of Equation (\ref{SrachDynamicFormulaeq}) and the empirical simulation results. It summarizes the RMSE, MSE, and absolute errors across various distributions and scenarios. The consistently low error values across all tested distributions confirm the robustness and accuracy of the proposed model. However, as shown in Table \ref{errorComparisontb}, the accuracy of the equation is influenced by both the degree distribution and network size, with Poisson-distributed networks exhibiting the highest deviations from the theoretical predictions.

\begin{table}[htbp]
\caption{Model accuracy across various degree distributions. The top section corresponds to a fixed network size of $N=500$ (as shown in Figure \ref{NormalDynamicCMPfigk}), while the bottom section presents results for a fixed average degree of $ \langle \hat k \rangle  = 5$ with varying network sizes N (as illustrated in Figure \ref{NormalDynamicCMPfigN}).}
\centering
\begin{tabular}{|l|c|c|c|c|}
\hline
\textbf{Scenario} & \textbf{Distribution} & \textbf{RMSE} & \textbf{MSE} & \textbf{Absolute Error} \\\hline
\multirow{3}{*}{Fixed $N$ = 500} & Normal & 0.67 & 0.45 & 0.45 \\\cline{2-5}
                                 & Pareto & 0.82 & 0.673 & 0.558 \\\cline{2-5}
                                 & Poisson & 0.88 & 0.775 & 0.775 \\\hline

\multirow{3}{*}{Fixed $\langle \hat k \rangle$ = 5} & Normal & 1.05 & 1.02 & 0.86 \\\cline{2-5}
                                              & Pareto & 1.45 & 2.114 & 1.277 \\\cline{2-5}
                                              & Poisson & 2.86 & 8.20 & 2.054 \\\hline

\end{tabular}

\label{errorComparisontb}
\end{table}

Figure \ref{dynamocFormulaKconstant} demonstrates that as the network size N increases while maintaining the same degree distribution, the effective diameter of the network also increases. This result aligns with theoretical expectations, as larger networks typically require more steps for flow to traverse from one side to another. On the other hand, Figure \ref{dynamicNewFormulaNConstantFig} illustrates that when the network size is fixed, an increase in the average degree $ \langle \hat k \rangle$ leads to a decrease in the effective diameter. This occurs because higher connectivity enhances the reachability of nodes, reducing the number of steps needed to cover the network. Together, these figures highlight the interplay between network size and connectivity in shaping the structural efficiency of information propagation.

\begin{figure}
    \centering
    \begin{subfigure}{7cm}
        \centering
        \includegraphics[width=7cm]{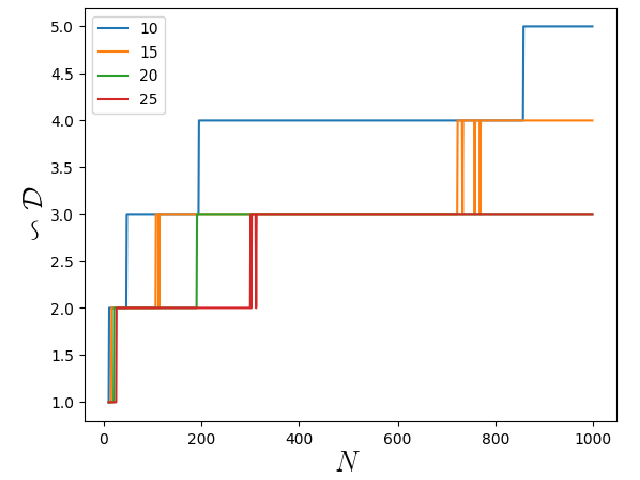}
        \caption{}\label{dynamocFormulaKconstant}
    \end{subfigure}
    \hfill
    \begin{subfigure}{8cm}
        \centering
        \includegraphics[width=8cm]{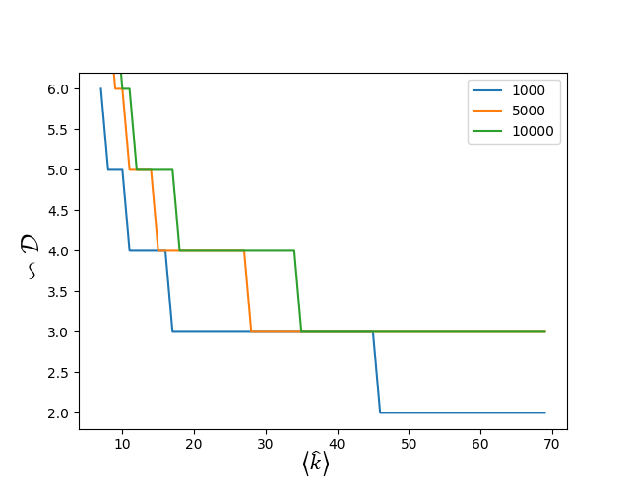}
        \caption{}\label{dynamicNewFormulaNConstantFig}
    \end{subfigure}
    \caption{Effect of network parameters on effective diameter ($\sim \mathcal{D}$). (a) The impact of network size $N$ on the effective diameter for different average degrees ($\langle \hat{k} \rangle \in \{10, 15, 20, 25\}$). (b) The effect of average degree on the effective diameter for networks of sizes $N \in \{1000, 5000, 10000\}$.}
\end{figure}



\subsection{Real Temporal Networks}

Following we analyse the diameter of networks in four different real data sets. The high school \cite{46journal.pone.0107878}, hospital \cite{47vanhems2013estimating}, work place \cite{45Genois2018} and contacts in a conference \cite{44/journal.pone.0011596}. We get all of them from the sociopattern website  (\url{http://www.sociopatterns.org}). These data sets represent the temporal connection among the participant. 

Figure \ref{fig:connection_duration} represent the connection duration among all nodes in the conference and high school. All the connections in the workplace and hospital last just one time step.  The connection duration follows a scale-free pattern meaning that there are lots of connections with low duration and some connections with high duration. Actually, in the real world we expect the same manner. 

\begin{figure}
  \begin{subfigure}{.5\textwidth}
    \centering\includegraphics[width=.8\linewidth]{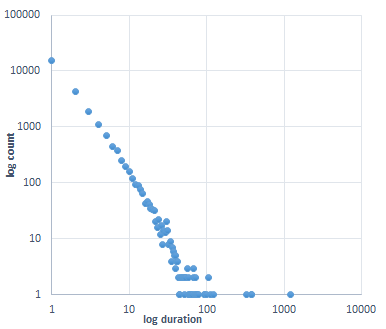}
    \caption{}
  \end{subfigure}
  \begin{subfigure}{.5\textwidth}
    \centering\includegraphics[width=.8\linewidth]{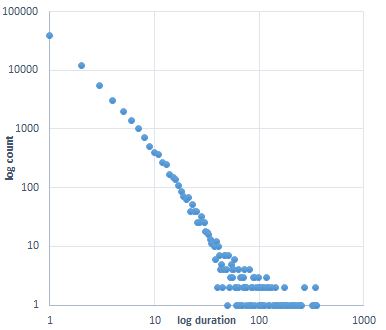}
    \caption{}\label{fig:highschool_contact_duration}
  \end{subfigure}

  \caption{Log-log plot of connection durations in (a) the Hospital contact network and (b) the High School contact network.}\label{fig:connection_duration}
\end{figure}

To be more precise in the matter of temporal network connections, we study the pattern of the gap between two connections, too. Figure \ref{fig:delay_duration} represents the gap between two connections among any two connected nodes. Also, the gap between connections follows a scale-free pattern, too. Actually, the gap better describes the nature of the time range between two disconnections make the connection period. 

\begin{figure}
  \begin{subfigure}{.5\textwidth}
    \centering\includegraphics[width=.8\linewidth]{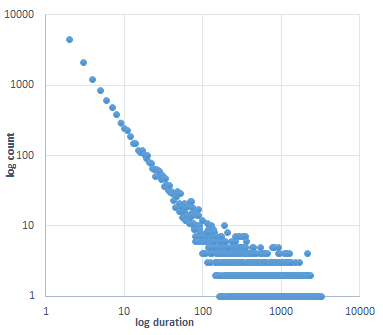}
    \caption{}\label{fig:confrence_delay_duration}
  \end{subfigure}
  \begin{subfigure}{.5\textwidth}
    \centering\includegraphics[width=.8\linewidth]{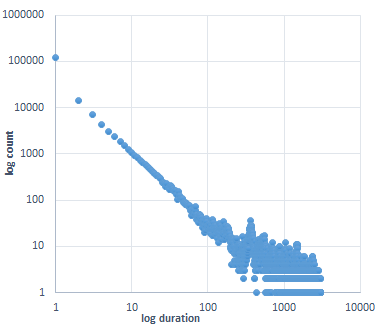}
    \caption{}\label{fig:highschool_delay_duration}
  \end{subfigure}
  \begin{subfigure}{.5\textwidth}
    \centering\includegraphics[width=.8\linewidth]{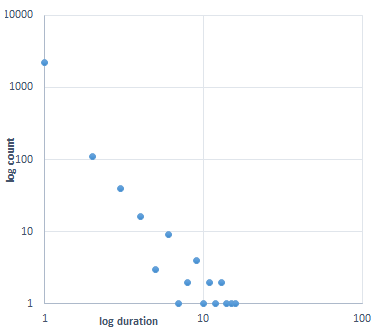}
    \caption{}\label{fig:workplace_delay_duration}
  \end{subfigure}
  \begin{subfigure}{.5\textwidth}
    \centering\includegraphics[width=.8\linewidth]{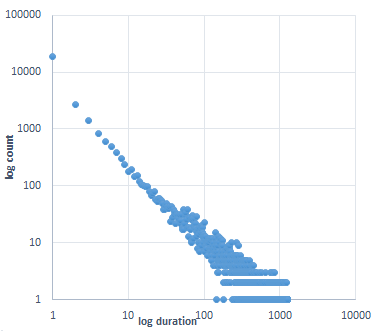}
    \caption{}\label{fig:hospital_delay_duration}
  \end{subfigure}
  \caption{Log-log plots of delay durations between consecutive connections in: (a) Conference, (b) High School, (c) Workplace, and (d) Hospital datasets.}\label{fig:delay_duration}
\end{figure}

Now, we study the feature of the networks in temporal scale-free networks. We will study three diameters in the network, the $\tau\mathcal{D}$ , $\backsim$$\mathcal{D}$ and $*\mathcal{D}$.

Figure \ref{fig:diameters} shows the comparison between three diameters in four data sets when $p\%$ of nodes in the network are removed. In the static networks, removing nodes causes the decrease in the diameter. But in a temporal network the concept of diameter is different, by removing the nodes there is a possibility that the diameter increases like figure \ref{fig:highschool_d}, have no changes like \ref{fig:confrence_d} and \ref{fig:hospital_d} or even decreases like in \ref{fig:workplace_d}. So the results show that there is no relation between the diameter in the static network and the diameter in the temporal networks.

Another important point in the plots of this figure is the relation between the $\tau\mathcal{D}$  and the peak diameter. In all the data sets, both $\tau \mathcal{D}$ and $*\mathcal{D}$ have close values representing that the $\tau \mathcal{D}$ is a crucial moment in the epidemic spreading and we expect that in a time very close to the $\tau\mathcal{D}$, the moment that the infected nodes are more than any other time, happens. So this can give us a prediction for controlling the epidemic. 

\begin{figure}[ht]
\begin{subfigure}{.5\textwidth}
  \centering
  \includegraphics[width=.8\linewidth]{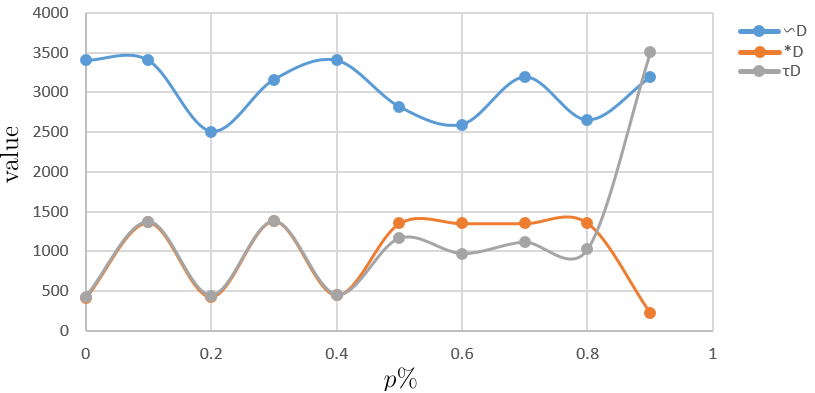}
  \caption{}
  \label{fig:confrence_d}
\end{subfigure}
\begin{subfigure}{.5\textwidth}
  \centering
  \includegraphics[width=.8\linewidth]{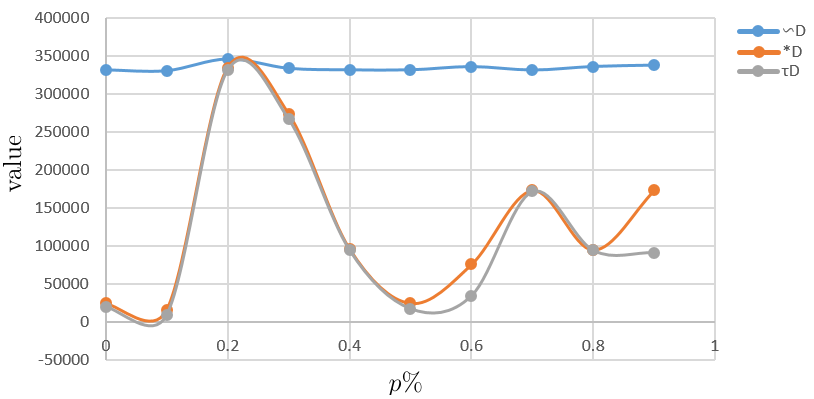}
  \caption{}
  \label{fig:hospital_d}
\end{subfigure}
\begin{subfigure}{.5\textwidth}
  \centering
  \includegraphics[width=.8\linewidth]{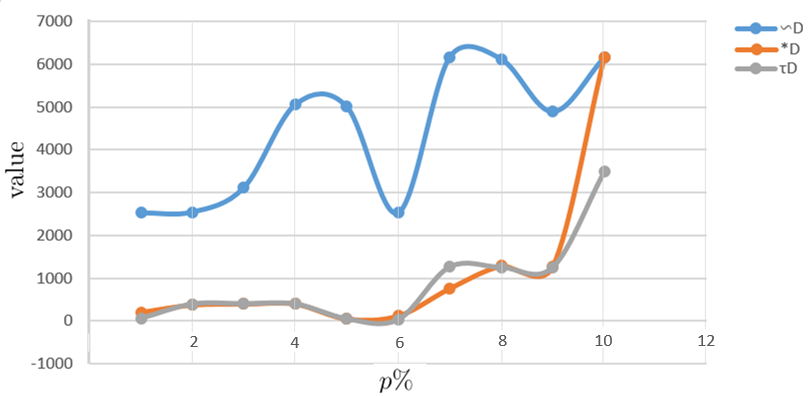}
  \caption{}
  \label{fig:highschool_d}
\end{subfigure}
\begin{subfigure}{.5\textwidth}
  \centering
  \includegraphics[width=.8\linewidth]{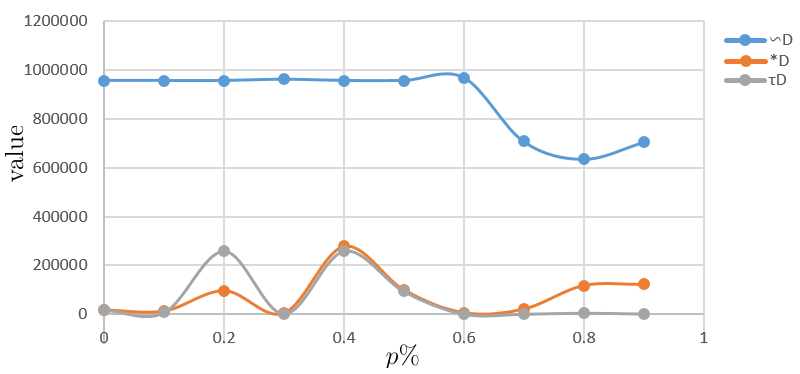}
  \caption{}
  \label{fig:workplace_d}
\end{subfigure}
\caption{Comparison of three types of diameters when $p\%$ of nodes are removed: $\tau\mathcal{D}$, $*\mathcal{D}$, and $\backsim \mathcal{D}$ in: a) the conference, b) the hospital, c) the high school, and d) the workplace network. }
\label{fig:diameters}
\end{figure}

Table \ref{feature_sf_tb} represents the relation between the number of nodes, edges, average degree and diameters in data sets. In all the data sets removing nodes, causes more significant changes in $\tau\mathcal{D}$ and $*\mathcal{D}$ than the $\backsim \mathcal{D}$. Consider a node with low connections that rarely connect to other nodes, as we expect this is the last node that gets infected, so, it does not have an effect on $\tau \mathcal{D}$ and $*\mathcal{D}$ but it has an effect on the $\backsim\mathcal{D}$. The $\backsim\mathcal{D}$ lasts until the last node connects to infected nodes. Since the probability of removing a specific node with a specific feature (the least connections with rare connection time) among all nodes, is low, thus the changes for the $\backsim\mathcal{D}$ are low.

\begin{longtable}{ |c|c|c|c|c|c|c|c| }
\caption{Impact of node removal on network structure across four datasets. For each proportion $p$, the table reports the number of nodes ($N$), edges ($E$), average degree ($\langle \hat k \rangle$), and three types of diameters: $\sim \mathcal{D}$, $\tau \mathcal{D}$ and $* \mathcal{D}$.}\label{feature_sf_tb} \\
\hline
 \textbf{Data set} & \textbf{p} & \textbf{N} & \textbf{E} & $\boldsymbol{\langle k\rangle}$ & $\boldsymbol{\sim \mathcal{D}}$ & $\boldsymbol{\tau \mathcal{D}}$ & $\boldsymbol{* \mathcal{D}}$  \\\hline
\multirow{10}{*}{{High school}}
  &1& 327 & 5818 & 17.79  & 2526 & 209 & 56    \\
  &0.1  & 296 & 4751 & 16.05  & 2543 & 382 & 394   \\
 &0.2  & 253 & 3519 & 13.90  & 3114 & 402 & 400   \\
 &0.3  & 244 & 3179 & 13.2   & 5054 & 403 & 397   \\
 &0.4  & 197 & 2203 & 11.18  & 5007 & 55  & 49    \\
 &0.5  & 165 & 1529 & 9.26   & 2526 & 126 & 30    \\
 &0.6  & 115 & 746  & 6.48   & 6163 & 764 & 1262  \\
 &0.7  & 100 & 538  & 5.38   & 6119 & 1289& 1252  \\
 &0.8  & 67  & 221  & 3.29   & 4889 & 1273 & 1252 \\
 & 0.9  & 29  & 40   & 1.37   & 6152 & 6152 & 3484 \\
\hline
\multirow{10}{*}{Conference}
  &1 & 403 & 9565 & 23.73   & 3408 & 421  & 410  \\
 & 0.1  & 365 & 7594 & 20.80   & 3408 & 1375 & 1364 \\
 & 0.2  & 312 & 5507 & 17.65   & 2507 & 442  & 429  \\
 & 0.3  & 277 & 4333 & 15.64   & 3163 & 1382 & 1383 \\
 & 0.4  & 241 & 3430 & 14.23   & 3408 & 451  & 450  \\
 &0.5  & 207 & 2588 & 12.50   & 2819 & 1163 & 1353 \\
 &0.6  & 146 & 1216 & 8.32    & 2596 & 970  & 1353 \\
 & 0.7  & 115 & 691  & 6.00    & 3197 & 1114 & 1353 \\
 & 0.8  & 80  & 353  & 4.41    & 2651 & 1025 & 1356 \\
 & 0.9  & 34  & 56   & 1.63    & 3197 & 3508 & 2234 \\
\hline
\multirow{10}{*}{Hospital}
&1 & 75 & 1139 & 15.18  & 331640 & 24700  & 20100 \\
 &0.1    & 66 & 846  & 12.81  & 330460 & 16100  & 9160 \\
 &0.2    & 60 & 696  & 11.6   & 346360 & 334340 & 331640 \\
 &0.3    & 49 & 504  & 10.28  & 333740 & 274380 & 267720 \\
&0.4    & 47 & 464  & 9.87   & 331640 & 96700  & 95140 \\
 &0.5    & 36 & 285  & 7.91   & 331640 & 25040  & 17800 \\
 &0.6    & 33 & 170  & 5.15   & 336040 & 75960  & 34780 \\
 &0.7    & 23 & 138  & 6.0    & 331440 & 173220 & 171920 \\
 &0.8    & 11 & 20   & 1.81   & 336220 & 94980  & 94320 \\
 &0.9    & 7  & 5    & 0.71   & 338020 & 173320 & 91620 \\
\hline
\multirow{10}{*}{workplace}
&1& 217 & 4247 & 19.69 & 955980 & 18100 & 17240 \\
 &0.1   & 197 & 3620 & 18.37 & 955980 & 13580 & 7960 \\
 &0.2   & 165 & 2433 & 14.74 & 955980 & 96680 & 259180\\
 &0.3   & 152 & 1983 & 13.04 & 961120 & 5000 & 1620 \\
 &0.4   & 135 & 1631 & 12.08 & 955980 & 277740 & 259180 \\
 &0.5   & 121 & 1391 & 11.49 & 955980 & 99340 & 93860 \\
 &0.6   & 87  & 654  & 7.51  & 966940 & 6660 & 1280 \\
 &0.7   & 65  & 399  & 6.13  & 707520 & 21440 & 1340 \\
 &0.8   & 51  & 252  & 4.94  & 636060 & 117760 & 5680 \\
 &0.9   & 23  & 59   & 2.56  & 704160 & 122180 & 2140 \\
\hline
\end{longtable}

Figure \ref{fig:heatmaps} presents Pearson correlation heatmaps for four distinct datasets, highlighting the relationships among key network parameters such as network size, edge count, average degree, and various types of diameter.

Across all four datasets, a strong positive correlation is observed among the three types of diameter, indicating that they tend to vary together. Additionally, the temporal average degree shows a direct relationship with both the number of temporal edges and the number of nodes — as the number of edges increases relative to the number of nodes, the average degree also increases.

Most notably, an inverse correlation exists between the average degree and all three types of diameter. This means that as the average degree decreases, the diameters increase, reflectng the intuitive notion that sparser networks tend to have longer paths between nodes.

\begin{figure*}[h]
    \centering
    \begin{subfigure}{0.45\textwidth}
        \centering
        \includegraphics[width=\textwidth]{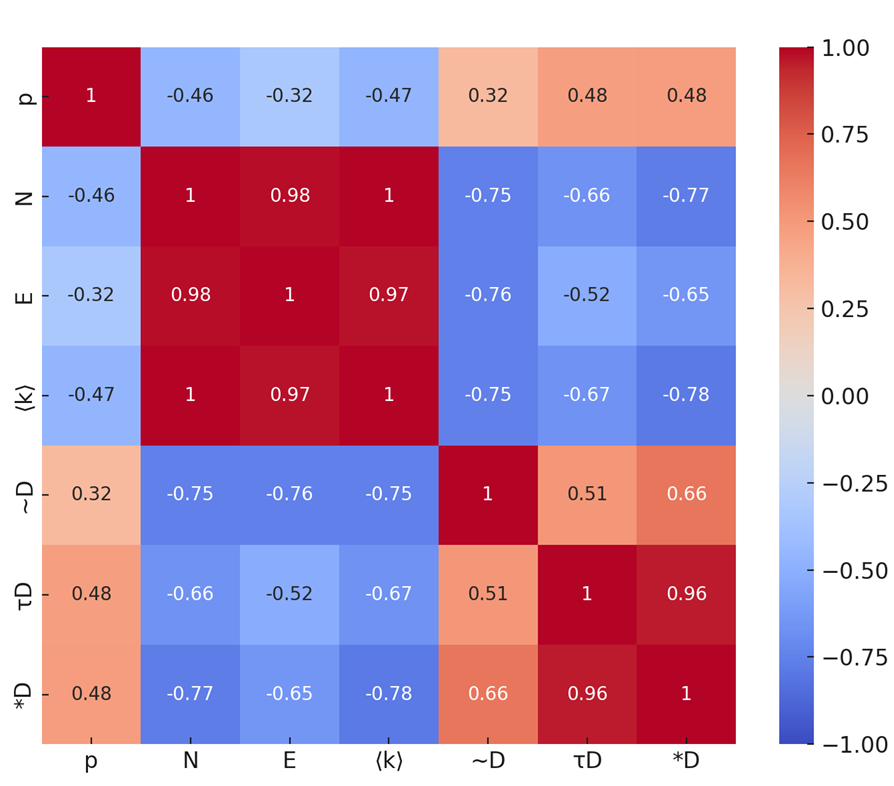}
        \caption{}
        \label{fig:highschool_heatmap}
    \end{subfigure}
    \hfill
    \begin{subfigure}{0.45\textwidth}
        \centering
        \includegraphics[width=\textwidth]{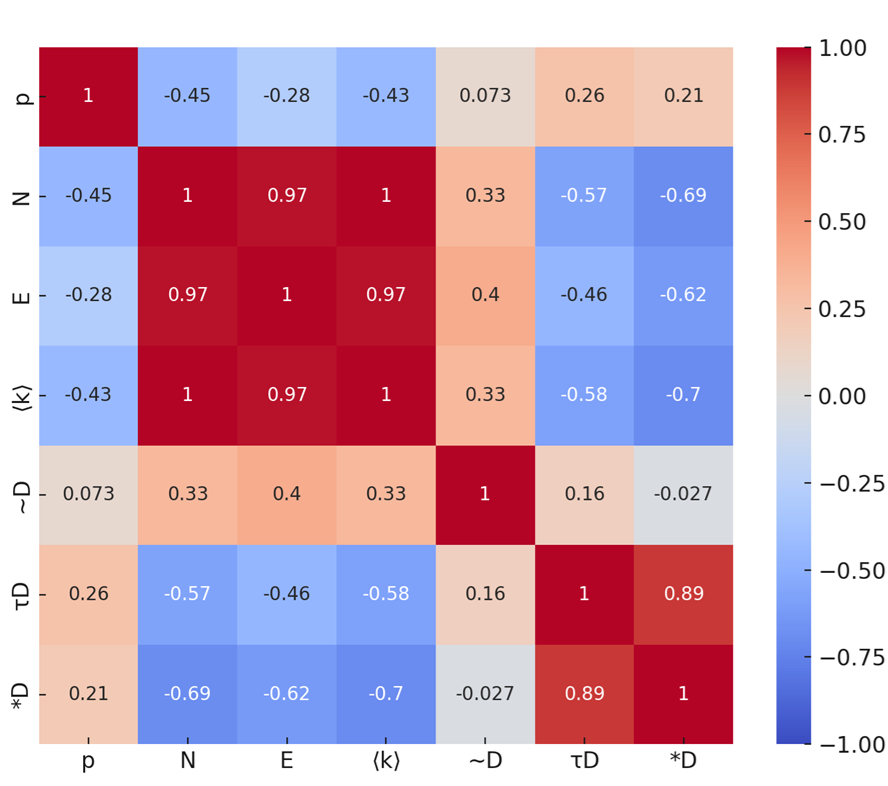}
        \caption{}
        \label{fig:conference_heatmap}
    \end{subfigure}
    
    \vspace{0.5cm}
    
    \begin{subfigure}{0.45\textwidth}
        \centering
        \includegraphics[width=\textwidth]{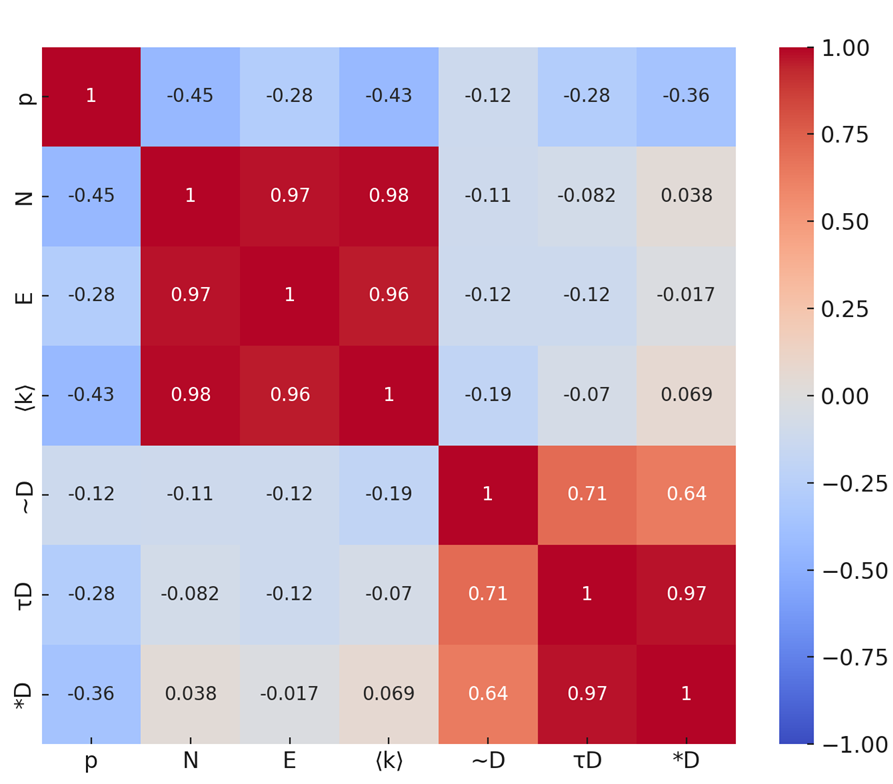}
        \caption{}
        \label{fig:hospital_heatmap}
    \end{subfigure}
    \hfill
    \begin{subfigure}{0.45\textwidth}
        \centering
        \includegraphics[width=\textwidth]{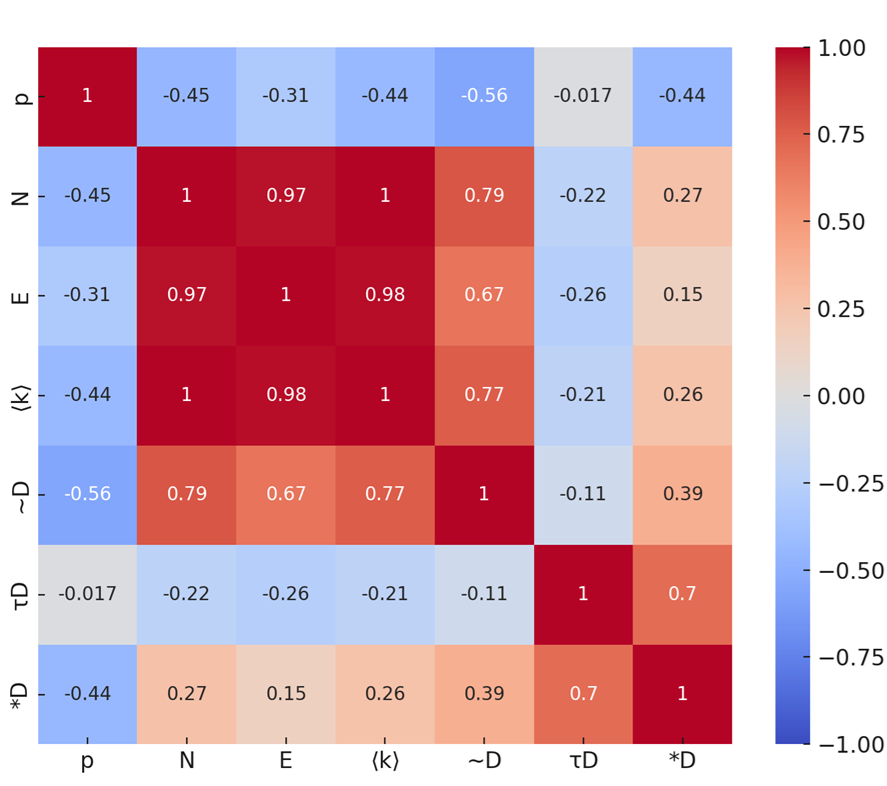}
        \caption{}
        \label{fig:workplace_heatmap}
    \end{subfigure}
    
    \caption{Pearson correlation heatmaps illustrating the relationships between key parameters across four datasets: (a) High School, (b) Conference, (c) Hospital, and (d) Workplace. The color scale ranges from -1 to 1, where red denotes strong positive correlations and blue represents strong negative correlations.}
    \label{fig:heatmaps}
\end{figure*}

\section{Conclusions and Future Directions}

In this study, we introduced a comprehensive mathematical framework to quantify the concept of diameter in temporal networks. Unlike traditional static approaches, our model incorporates time-dependent connectivity and captures the dynamic nature of real-world interactions. Central to our framework are new definitions such as the Effective Diameter ($\backsim \mathcal{D}$), Peak Diameter ($*$$\mathcal{D}$), and $\tau$-Diameter ($\tau$$\mathcal{D}$), which collectively offer a nuanced understanding of temporal reachability and information flow.

We validated our model through extensive simulations across networks of varying sizes and degree distributions. The results revealed consistent trends, specifically, a decrease in effective diameter with increased average degree and an increase in diameter with larger network sizes. These outcomes confirm our theoretical expectations and demonstrate the applicability of our framework in modeling real-world dynamics.

To evaluate the accuracy of our model, we conducted quantitative comparisons using RMSE, MSE, and absolute error metrics. The low error margins underscore the robustness of our approach, though deviations in networks with Poisson-like degree distributions suggest opportunities for refinement in such contexts.

We further applied our framework to four empirical temporal networks: high school, hospital, workplace, and conference datasets, to investigate the effect of node removal on different diameter metrics. The results showed that  $\tau$$\mathcal{D}$ and peak diameter are more sensitive to node removal than effective diameter, indicating the complex interplay between temporal structure and network resilience.

Our findings lay the groundwork for several promising directions for future research. One potential extension involves integrating more sophisticated temporal models that capture bursty behavior and heterogeneous interaction patterns. Additionally, incorporating node and edge attributes, such as roles or weights, could provide deeper insights into how structural properties influence temporal connectivity. Finally, applying this framework to domain-specific problems, such as epidemic modeling, cybersecurity, or mobility networks, could yield practical strategies for improving resilience and efficiency in time-sensitive systems.

\bibliographystyle{plain}

\section{Appendix}\label{appendix}
The algorithm for Dynamic simulation

\begin{algorithm}
\caption{Simulation of procedure to find diameter in teporal networks}\label{alg:capDynamic}
\begin{algorithmic}[1]
\Procedure{FindDiameterDynamic}{}
\State \textbf{visited\underline{\space}step}: The step at which the flow reaches the node
\State $\textbf{step} \gets 0$: counter of steps
\State \textbf{visited} [\space] : an array of visited nodes
\State\textbf{active\underline{\space}steps} [\space,\space] $_{EdgeCount\times MaximumSteps}$: 2d binary array
defining the steps a link is active. Each row is generated based on the main distribution.
\State $\textbf{diameter} \gets 0$
\State \textbf{reachableSet}[\space]: an array of nodes visited in each step. Reachable set of that step
\State \textbf{temp\underline{\space}reachableSet}[\space]: temporary array to keep the newly visited nodes

\While{$legth(visited) < NodeCount$}
    \State diameter += 1
   \For{$v_i$ in reachableSet}
        \For {$ne_j$ in $v_i$.neighbors()}
            \State edgeId = edge($v_i$,$ne_j$)
            \If {active\underline{\space}steps[edgeId,step] == 1 \textbf{and} $v_i.visitedStep$ $<$ step \textbf{and}
            $ne_j$ \textbf{not in} visited}
                \State visited.add($ne_j$)
                \State $ne_j$.visitedStep = step
                \State temp\underline{\space}reachableSet.add($ne_j$)

            \EndIf
        \EndFor
   \EndFor
   \State visited += reachableSet
   \State reachableSet = temp\underline{\space}reachableSet
   \State step +=1

\EndWhile
\State \textbf{return} diameter
\EndProcedure
\end{algorithmic}
\end{algorithm}
\end{document}